\documentclass[12pt]{article}
\usepackage{amsmath, amssymb, amsthm}

\def\0{{\mathbf 0}}
\def\1{{\mathbf 1}}

\def\F{{\mathbb F}}

\def\Z{{\mathbb Z}}

\def\RR{{\mathcal R}}

\def\Zp8{{\Z_{p^\infty}}}

\newtheorem{thm}{Theorem}[section]
\newtheorem{prop}[thm]{Proposition}
\newtheorem{lem}[thm]{Lemma}
\newtheorem{cor}[thm]{Corollary}
\theoremstyle{definition}
\newtheorem{defn}[thm]{Definition}

\newtheorem{ex}[thm]{Example}

\date{}

\begin{document}

\newcommand{\comment}[1]{} 

\title{Computational Results of \\ Duadic Double Circulant Codes}

\author{
Sunghyu Han \\
School of Liberal Arts \\
Korea University of Technology and Education \\
 Cheonan 330-708, S. Korea \\
 {Email: \tt sunghyu@kut.ac.kr} \\
and \\
  \\
Jon-Lark Kim \thanks {corresponding author}\\
 Department of Mathematics \\
 University of Louisville \\
  Louisville, KY 40292, USA\\
{Email: \tt jl.kim@louisville.edu} \\
} \maketitle

\begin{abstract}
Quadratic residue codes have been one of the most important
classes of algebraic codes. They have been generalized into duadic
codes and quadratic double circulant codes. In this paper we
introduce a new subclass of double circulant codes, called
{\em{duadic double circulant codes}}, which is a generalization of
quadratic double circulant codes for prime lengths. This class generates optimal self-dual codes, optimal linear codes, and linear codes with the best known parameters in a systematic way. We describe a method to construct duadic double circulant codes using
$4$-cyclotomic cosets and give certain duadic double circulant codes over $\mathbb F_2, \mathbb F_3, \mathbb F_4,
\mathbb F_5$, and $\mathbb F_7$. In particular, we find a new
ternary self-dual $[76,38,18]$ code and easily rediscover  optimal
binary self-dual codes with parameters $[66,33,12]$, $[68,34,12]$,
$[86,43,16]$, and $[88,44,16]$ as well as a formally self-dual
binary $[82,41,14]$ code.

\end{abstract}




\newpage

\section{Introduction}

Quadratic Residue (QR) codes have been one of the most interesting
classes of linear codes. They have a good decoding algorithm
called permutation decoding, thanks to a large automorphism
group. They also have high minimum distances and sometimes are
self-dual.

On the other hand, M. Karlin~\cite{Kar} considered binary double
circulant codes based on quadratic residues. Later, V.
Pless~\cite{Ple72} considered ternary double circulant codes based
on quadratic residues, producing the famous Pless symmetry codes.
She also introduced a class of duadic codes as a natural
generalization of QR codes. Duadic codes were generalized
into triadic codes~\cite{PleRus}, polyadic codes~(\cite{BruPle},~\cite{LinXin}) and duadic codes~\cite{LanSol} over $\mathbb
Z_4$. P. Gaborit~\cite{Gab} introduced quadratic double
circulant codes which included Karlin's above construction and the
Pless symmetry codes. He constructed new infinite families of self-dual codes over $GF(4), GF(5), GF(7)$, and $GF(9)$~\cite{Gab}. A uniform generalization of quadratic
double circulant codes was given by S. Dougherty, et.
al.~\cite{DouKimSol} using two class association schemes.

\medskip

In this paper, we introduce a new subclass of double circulant
codes, called {\em duadic double circulant codes}. This is a
generalization of quadratic double circulant codes for prime
lengths and also can produce non self-dual codes.  This class generates optimal self-dual codes, optimal linear codes, and linear codes with the best known parameters in a systematic way.
For another motivation of duadic double circulant codes, we observe that there has been a classification of extremal double circulant self-dual codes of lengths up to $88$ (see~\cite{GulHar},~\cite{GulHar2006}). These codes are classified by an exhaustive computer search. Hence it is natural to construct a subclass of double circulant codes algebraically which contains both self-dual and non self-dual codes.

We compute all duadic double circulant codes of lengths up to $82, 82, 62, 58, 38$ over $\mathbb F_2, \mathbb F_3, \mathbb F_4,
\mathbb F_5$, and $\mathbb F_7$, respectively, that are based on $4$-cyclotomic
cosets. In particular, we rediscover optimal binary self-dual
codes with parameters $[66,33,12]$, $[68,34,12]$, $[86,43,16]$,
and $[88,44,16]$ as well as a formally self-dual binary
$[82,41,14]$ code. Using $9$-cyclotomic cosets, we find a new
ternary self-dual $[76,38,18]$ code. We further construct optimal
formally self-dual codes or codes with the best known parameters,
which cannot be obtained from quadratic double circulant construction (see
Example~\ref{ex:pure_bord} (iv) and Example~\ref{ex:n=43}).



\section{Splitting}

We begin with some definitions.
Let $n$ be a positive integer and $a$ ($1 \le a \le n-1$) be an
integer such that $\mbox{gcd}(a,n) = 1$. Then the function $\mu_a$
defined on $\{0, 1, 2, \cdots, n-1\}$ by $i\mu_a \equiv ia \pmod
n$ is called a \textit{multiplier}.

\begin{defn} \label{def:spl}
Let $n$ be an odd positive integer with $n>1$. A pair $(S_1, S_2)$ of two sets
$S_1$ and $S_2$ is called a {\em (generalized duadic) splitting}
of $n$ if the following two conditions are satisfied:
\begin{enumerate}
\item $S_1$ and $S_2$ satisfy
\begin{equation*}
S_1 \cup S_2 = \{ 1, 2, \cdots, n-1 \} ~~ \textrm{and} ~~ S_1 \cap
S_2 = \emptyset,
\end{equation*}
\item there is a multiplier $\mu_a$ such that
\begin{equation*}
S_1\mu_a = S_2 ~~ \textrm{and} ~~ S_2\mu_a = S_1.
\end{equation*}
\end{enumerate}
\end{defn}

In fact, condition (ii) of Definition~\ref{def:spl} can be
weakened as follows.

\begin{lem} \label{lem:multiplier}
Let $A=\{1, 2, \cdots, n-1 \}, {\mbox{ with odd }}n\geq 1$, $S_1
\cup S_2 = A$, and $S_1 \cap S_2=\emptyset$. Let $\mu_a$ be a
multiplier. If $S_1\mu_a = S_2$, then $S_2\mu_a =
S_1$.
\end{lem}
\begin{proof}
Since $\mu_a$ is bijective on $A$, we have $A=A \mu_a = S_1 \mu_a
\cup S_2 \mu_a$ and $S_1 \mu_a \cap S_2 \mu_a=\emptyset$. Hence if
$S_1 \mu_a=S_2$, then $S_2\mu_a = S_1$.
\end{proof}

It is natural to ask when there exists a splitting of a given odd integer $n$.

\begin{prop} \label{prop:spl}
Let $ord_n(a)$ denote the multiplicative order of $a$ modulo $n$. Then the
following hold.
\begin{enumerate}
\item There exists a splitting for any odd $n$ with $n>1$. \label{item:01}
\item If $ord_n(a)$ is odd, then there is no splitting of $n$ with $\mu_a$. \label{item:02}
\item If $n$ is a prime and $ord_n(a)$ is $n-1$,
then the splitting $S_1$ and $S_2$ is the set of quadratic residues and
the set of quadratic nonresidues of $n$. \label{item:03}
\end{enumerate}
\end{prop}
\begin{proof}
For (\ref {item:01}), we can take $S_1= \{1, 2, \cdots,
\frac{n-1}{2} \}$ and $S_2= \{\frac{n+1}{2}, \frac{n+3}{2},
\cdots, n-1 \}$ with $\mu_{-1}$. For (\ref {item:02}) and (\ref
{item:03}), without loss of generality, we may assume that $S_1$
contains $1$. Then $S_1=\{1, a^2, a^4, \cdots\}$ and $S_2=\{a,
a^3, a^5, \cdots\}$ under $\mu_a$ by Lemma~\ref{lem:multiplier}. Suppose
$ord_n(a) = 2k+1$ for some $k$. Then $a^{2k} \in S_1$ and
$a^{2k}\mu_a = 1 \in S_1$. This is a contradiction. This proves (ii). Next suppose that $n$ is a prime and that
$ord_n(a) = n-1$. Then $S_1$ is the set of quadratic residues and
$S_2$ is the set of quadratic nonresidues of $n$, proving (iii).
\end{proof}

\bigskip

\section{Duadic double circulant codes}

Now we are ready to define Duadic Double Circulant (DDC) codes.
First, we choose a splitting $(S_1, S_2)$ for an odd positive
integer $n$, where $S_1$ and $S_2$ do not need to be unions of nonzero $q$-cyclotomic cosets.
 With this splitting, we define two generating
matrices below.
\begin{eqnarray*}
P_n(r,s,t, S_1, S_2) &=& (I | D_n(r,s,t, S_1, S_2)), \\
B_n(\alpha, \beta, \gamma, r,s,t, S_1, S_2) &=& \left(
  \begin{array}{c|ccc|c|ccc}
    1 & 0 & \cdots & 0 & \alpha & \beta & \cdots & \beta \\
    \hline
    0 &  &  &  & \gamma & & & \\
    \vdots &  & I &  & \vdots & &D_n(r,s,t, S_1, S_2)& \\
    0 &  &  &  & \gamma & && \\
  \end{array}
\right),
\end{eqnarray*}
where $r,s,t,\alpha, \beta, \gamma \in \F_q$, $I$ is an $n \times
n$ identity matrix, and $D_n(r,s,t, S_1, S_2)$ is an $n \times n$ circulant
matrix whose first row is defined as follows. Suppose that $(a_0,
a_1, \cdots, a_{n-1})$ is the first row of $D_n(r,s,t, S_1, S_2)$. Then
define $a_0 = r$, $a_i = s$ if $i \in S_1$, and $a_i = t$ if $i
\in S_2$. Using these two generating matrices, we define $[2n, n]$
and $[2n+2, n+1]$ linear codes over $\F_q$ and we denote them by
$\mathcal{P}_n(r, s, t, S_1, S_2)$ and $\mathcal{B}_n(\alpha, \beta, \gamma,
r, s, t, S_1, S_2)$ respectively. We call $\mathcal{P}_n(r, s, t, S_1, S_2)$ a duadic
pure double circulant code, $\mathcal{B}_n(\alpha, \beta, \gamma,
r, s, t, S_1, S_2)$ duadic bordered double circulant code, and we call both
\textit{duadic double circulant codes}. Note that duadic double
circulant codes contain Gaborit's quadratic double circulant codes
\cite{Gab} when $n$ is a prime by Proposition \ref{prop:spl}
(\ref{item:03}). We also remark that $\mathcal{P}_n(r, s, t, S_1, S_2)$ over
$\mathbb F_2$ is always formally self-dual~\cite[Sec. 9.8]{HufPle}, while
$\mathcal{B}_n(\alpha, \beta, \gamma, r, s, t, S_1, S_2)$ is not
necessarily. (See Example~\ref{ex:pure_bord} in
Section~\ref{sec:comp_res}.)

\begin{prop}
The codes  $\mathcal{P}_n(r, s, t, S_1, S_2)$ and
$\mathcal{P}_n(r, s, t, S_2, S_1)$ are equivalent. Similarly,
$\mathcal{B}_n(\alpha, \beta, \gamma,
r, s, t, S_1, S_2)$ and $\mathcal{B}_n(\alpha, \beta, \gamma,
r, s, t, S_2, S_1)$ are equivalent.
\end{prop}

The proof of this proposition follows from the below lemmas.

\begin{lem}
Let $\F_q$ be an arbitrary finite field. Let $\RR_n = \F_q[x]/(x^n
-1), \gcd(n,q)=1$ and $\mu_a$ be a multiplier with $\gcd(a,n)=1$.
Consider the following two pure double circulant codes in $\RR_n
\oplus \RR_n$.
\begin{eqnarray*}
C_1 &=& \{(f(x), e_1(x)f(x)) ~|~ f(x) \in \RR_n  \},  \\
C_2 &=& \{(f(x), e_2(x)f(x)) ~|~ f(x) \in \RR_n  \},
\end{eqnarray*}
where $e_2(x) = e_1(x)\mu_a$. Then $C_1$ and $C_2$ are permutation
equivalent.
\end{lem}
\begin{proof}
Let
\begin{equation*}
C_1\mu_a = \{(f(x)\mu_a, (e_1(x)f(x))\mu_a) ~|~ f(x) \in \RR_n  \}.
\end{equation*}
In \cite[Theorem 4.3.12]{HufPle}, we know that $\mu_a$ is an
automorphism of $\RR_n$. Therefore,
\begin{eqnarray*}
C_1\mu_a &=& \{(f(x)\mu_a, (e_1(x)\mu_a)(f(x)\mu_a) ) ~|~ f(x) \in \RR_n  \} \\
         &=& \{(h(x), e_2(x)h(x) ) ~|~ h(x) \in \RR_n  \} \\
         &=& C_2.
\end{eqnarray*}
\end{proof}
\begin{lem}
Let $\F_q$ be an arbitrary finite field. Let $\RR_n = \F_q[x]/(x^n
-1), \gcd(n,q)=1$ and $\mu_a$ be a multiplier with $\gcd(a,n)=1$.
Consider the following two bordered double circulant codes in
$\F_q \oplus \RR_n \oplus \F_q \oplus \RR_n$.
\begin{eqnarray*}
C_1 &=& \{(b, f(x), b\alpha + f(1)\gamma, b\beta j(x) + e_1(x)f(x)) ~|~ b\in \F_q, f(x) \in \RR_n  \},  \\
C_2 &=& \{(b, f(x), b\alpha + f(1)\gamma, b\beta j(x) + e_2(x)f(x)) ~|~ b\in \F_q, f(x) \in \RR_n  \},  \\
\end{eqnarray*}
where $\alpha, \beta, \gamma \in \F_q, e_2(x) = e_1(x)\mu_a, j(x) = 1+x+x^2 +\cdots +x^{n-1}
$. Then $C_1$ and $C_2$ are
permutation equivalent.
\end{lem}
\begin{proof}
Let
\begin{equation*}
C_1\mu_a = \{(b, f(x)\mu_a, b\alpha + f(1)\gamma, (b\beta j(x) + e_1(x)f(x))\mu_a ~|~ b\in \F_q, f(x) \in \RR_n  \}.  \\
\end{equation*}
Then,
\begin{eqnarray*}
C_1\mu_a &=& \{(b, f(x)\mu_a, b\alpha + f(1)\gamma, b\beta j(x) + (e_1(x)\mu_a)(f(x)\mu_a) ~|~ b\in \F_q, f(x) \in \RR_n  \}.  \\
         &=& \{(b, h(x), b\alpha + h(1)\gamma, b\beta j(x) + e_2(x) h(x) ~|~ b\in \F_q, h(x) \in \RR_n  \}.  \\
         &=& C_2.
\end{eqnarray*}
\end{proof}

\begin{defn} \label{def:duadic}
Suppose that there is a splitting $(S_1, S_2)$ of $n$ and let $q$ be a power of a prime with $\gcd(q,
n) = 1$. Furthermore, if $S_1$ and $S_2$ are unions of nonzero $q$-cyclotomic
cosets, then the cyclic codes with defining sets $S_1$ and $S_2$
are called the {\em{odd-like duadic codes of length $n$ over
$\mathbb F_q$.}} On the other hand, the cyclic codes with defining
sets $S_1 \cup \{ 0 \}$ and $S_2 \cup \{ 0 \}$ are called the
{\em{even-like duadic codes of length $n$ over $\mathbb F_q$.}}
\end{defn}

\begin{thm} {\em \cite[Chapter 6]{HufPle}} \label{thm:dua_exist}
Duadic codes of odd length $n$ over $\mathbb F_q$ exist if and only
if $q$ is a square modulo $n$.
\end{thm}

Thus we have a systematic way to construct splittings of $n$ using
$q$-cyclotomic cosets of $n$. This method also includes splittings
consisting of quadratic residues and nonresidues as shown below.

\begin{cor} \label{cor:cyc}
\begin{enumerate}
\item If $q$ is a square modulo $n$, then there is a splitting $(S_1, S_2)$, each of which consists of a
union of $q$-cyclotomic cosets of $n$. \label{cor:cyc1}

\item If $n$ is an odd prime and $q$ is a square modulo $n$, then
the set of quadratic residues of $n$ consists of a union of
$q$-cyclotomic cosets of $n$, and so does the set of quadratic
nonresidues of $n$.
 \label{cor:cyc2}
\end{enumerate}
\end{cor}
\begin{proof}
The existence of a splitting $(S_1, S_2)$, each of which
consists of a union of $q$-cyclotomic cosets of $n$, is equivalent
to the existence of duadic codes of length $n$ over $\F_q$. Part
(\ref{cor:cyc1}) follows from Theorem~\ref{thm:dua_exist}. Part
(\ref{cor:cyc2}) is an easy exercise in \cite[p. 237]{HufPle}.
\end{proof}

In what follows, we consider $4$-cyclotomic cosets of an odd $n$
since $4=2^2$ is a square modulo $n$.
 By Corollary \ref{cor:cyc} (ii), duadic double circulant codes with a
splitting $(S_1, S_2)$, each of which consists of a union of
$4$-cyclotomic cosets of $n$ contain Gaborit's quadratic double
circulant codes \cite{Gab} if $n>2$ is a prime.

\medskip

The total number of splittings which consist of unions of
$4$-cyclotomic cosets for an odd $n$ such that $3 \leq n \leq 41$
can be obtained from \cite[Table 6.1, Table 6.3]{HufPle} and
\cite{Ple}. We compute explicitly splittings for $3 \leq n
\leq 41$. We also compute all the splittings
which consist of unions of $9$-cyclotomic cosets for an odd $n$
such that $5 \leq n \leq 37$. To save space, we give them in~\cite{KimWeb}.

\medskip




\section{Computational results} \label{sec:comp_res}
We compute all DDC codes of lengths up to $82, 82, 62, 58, 38$ over $\mathbb F_2,
\mathbb F_3, \mathbb F_4, \mathbb F_5$, and $\mathbb F_7$, respectively, based on
$4$-cyclotomic cosets, and one DDC code over $\mathbb F_3$ based on $9$-cyclotomic cosets. The results are displayed in Table
\ref{tab:F2} to Table \ref{tab:F7}. Full tables including all the missing lengths can be found in~\cite{KimWeb}.
We have used Magma~\cite{Mag} whenever it is necessary.

In Table \ref{tab:F2}, the first column $n$
indicates an odd positive integer in the duadic splitting. The
second column ``cl'' indicates the corresponding code length. The
third column SD(I) indicates the maximum minimum distance for
self-dual Type I codes in our calculation. The fourth column
SD(II) indicates the maximum minimum distance for self-dual Type
II codes in our calculation. The fifth column NSD indicates the
maximum minimum distance for non self-dual codes in our
calculation. The sixth column O.SD(I) indicates the optimal
minimum distance for self-dual Type I codes in \cite{GabOtm}. The
seventh column O.SD(II) indicates the optimal minimum distance for
self-dual Type II codes in \cite{GabOtm}. The eighth column
O.Linear indicates the optimal minimum distance for linear codes
in \cite{codetables}. In the ninth column Comment, O(I) indicates
our construction for self-dual Type I code is optimal, O(II)
indicates our construction for self-dual Type II code is optimal,
O(L) indicates our construction is optimal in linear codes, and
B(L) indicates the maximum minimum distance of our construction is
equal to that of the best known linear codes.

In Table \ref{tab:F3}, O.SD comes from \cite{GabOtm}.
In Table \ref{tab:F4}, O.SD(E) and O.SD(H) come from \cite{GabOtm}, \cite{GulHar3}.
In Table \ref{tab:F5}, O.SD comes from \cite{GabOtm}, \cite{GulHar3}, \cite{HanKim}, \cite{HarMun}.
In Table \ref{tab:F7}, O.SD comes from \cite{GabOtm}, \cite{GulHar3}, \cite{GulHarMiy}.
For all tables, O.Linear comes from \cite{codetables}. From our DDC construction, we show that there are many optimal
self-dual codes or optimal linear codes (or codes which has the
best known linear code parameters).

\medskip

The following examples help readers to understand our construction method.
Furthermore, the examples are nontrivial to obtain and connect our codes with other known codes.

\begin{ex} \label{ex:pure_bord}



\begin{enumerate}
\item $n=15$ in $\F_2$. \\
With the following splitting
\begin{equation*}
S_1=\{1, 4, 3, 12, 7, 13, 5\}, \mu_{2}
\end{equation*}
$\mathcal{P}_{15}(0, 0, 1, S_1, S_2)$ is a $[30, 15, 8]$ optimal linear code
and $\mathcal{B}_{15}( 0, 1, 0, 0, 0, 1, S_1, S_2 )$ is a $[32, 16, 8]$
optimal linear code.

$\mathcal{P}_{15}(0, 0, 1, S_1, S_2)$ is an extremal formally self-dual even
code. There are exactly six $[30, 15, 8]$ extremal double
circulant formally self-dual even codes \cite{GulHar2}. Since the
order of automorphism group for our code is $60$, our code is
equivalent to $C_{30, 4}$ in \cite{GulHar2}. $\mathcal{B}_{15}( 0,
1, 0, 0, 0, 1, S_1, S_2 )$ is not a formally self-dual code.

\item $n=17$ in $\F_2$. \\
With the following spitting
\begin{equation*}
S_1=\{1, 4, 16, 13, 3, 12, 14, 5\}, \mu_{-2}
\end{equation*}
$\mathcal{P}_{17}(1, 0, 1, S_1, S_2)$ is a $[34, 17, 8]$ optimal linear code
and $\mathcal{B}_{17}( 0, 1, 0, 1, 0, 1, S_1, S_2 )$ is a $[36, 18, 8]$
optimal linear codes. We have checked that with the splitting of quadratic
residues and quadratic nonresidues of $n=17$, all the minimum
distance of DDC codes are less than $8$.

$\mathcal{P}_{17}(1, 0, 1, S_1, S_2)$ is a near-extremal formally self-dual even code.
There are exactly five weight distributions for
$[34, 17, 8]$ near-extremal double circulant
formally self-dual even codes \cite{GulHar2}.
Our code has $\alpha = 17$ in the notation of weight distribution \cite{GulHar2}.
$\mathcal{B}_{17}( 0, 1, 0, 1, 0, 1, S_1, S_2 )$ is not a formally self-dual code.

\item $n=33$ in $\F_2$. \\
With the following splitting
\begin{equation*}
S_1=\{1, 4, 16, 31, 25, 3, 12, 15, 27, 9, 7, 28, 13, 19, 10, 11\}, \mu_{-1},
\end{equation*}
we have that $\mathcal{P}_{33}(1, 0, 1, S_1, S_2)$ is a $[66, 33, 12]$ Type I optimal
self-dual code and that $\mathcal{B}_{33}(0, 1, 1, 0, 0, 1, S_1, S_2)$ is a $[68,
34, 12]$ Type I optimal self-dual code. We also show that $\mathcal{P}_{33}(0, 0,
1, S_1, S_2)$ has the best known parameters $[66, 33, 12]$.

From \cite{GulHar}, there are $3$ inequivalent pure double circulant
optimal $[66, 33, 12]$ Type I self-dual codes and there are $84$ inequivalent bordered
double circulant optimal $[68, 34, 12]$ Type I self-dual codes.
Both $\mathcal{P}_{33}(1, 0, 1, S_1, S_2)$ and $\mathcal{B}_{33}(0, 1, 1, 0, 0, 1, S_1, S_2)$
have automorphism group order $330$ and the number of
minimum codewords $858$ so that $\mathcal{P}_{33}(1, 0, 1, S_1, S_2)$ is
equivalent to $C_{66,21}$ in \cite[Table 3]{GulHar} and $\mathcal{B}_{33}(0, 1, 1, 0, 0, 1, S_1, S_2)$ is equivalent to $C'_{68,1}$ in \cite[Table 7]{GulHar}.
$\mathcal{P}_{33}(0, 0, 1, S_1, S_2)$ is a formally self-dual odd code.

\item $n=41$ in $\F_2$. \\
With the following splitting
\begin{equation*}
S_1=\{1, 4, 16, 23, 10, 40, 37, 25, 18, 31, 3, 12, 7, 28, 30, 38, 29, 34, 13, 11\}, \mu_{-2}
\end{equation*}
$\mathcal{P}_{41}(1, 0, 1, S_1, S_2)$ has the best known linear code
parameters $[82, 41, 14]$. We have checked that with the splitting of quadratic
residue and quadratic nonresidue of $n=17$, all the minimum
distance of duadic pure double circulant codes are less than
$14$. We recall that $\mathcal{P}_{41}(1, 0, 1, S_1, S_2)$ is a formally
self-dual even code.

\end{enumerate}
\end{ex}

\begin{ex} \label{ex:n=43}
$n = 43$ in $\F_2$. Consider the following splitting
\begin{equation*}
S_1:=\{1,4,16,21,41,35,11,3,12,5,20,37,19,33,7,28,26,18,29,30,34\},
\mu_{2},
\end{equation*}

Then $\mathcal{P}_{43}(0, 1, 0, S_1, S_2)$ is a $[86, 43, 16]$
optimal Type I self-dual code and the bordered code
$\mathcal{B}_{43}(0, 1, 1, 1, 1,
0, S_1, S_2)$ is a $[88, 44, 16]$ optimal Type II self-dual code. We have
checked that using the splitting of quadratic residues and
quadratic nonresidues $\mathcal{P}_{43}(0, 1, 0, S_1, S_2)$ is a $[86, 43,
14]$ Type I self-dual code and $\mathcal{B}_{43}(0, 1, 1, 1, 1,
0, S_1, S_2)$ is a $[88, 44, 16]$ optimal Type II self-dual code \cite{Gab}.
We have verified that our $\mathcal{P}_{43}(0, 1, 0, S_1, S_2)$ is
equivalent to the code in \cite{DouGulHar}, and
$\mathcal{B}_{43}(0, 1, 1, 1, 1, 0, S_1, S_2)$ is not equivalent to the code
$\mathcal{B}_{43}(0, 1, 1, 1, 1, 0, S_1, S_2)$ with the splitting of
quadratic residue and quadratic nonresidue. Note that there exist
at least 70 extremal binary doubly-even self-dual codes of length
88 \cite{GooYor}. On the other hand, it is known~\cite{GulHar2006} that there are exactly $151$ extremal binary doubly-even self-dual codes of length 88.
 Hence our code $\mathcal{B}_{43}(0, 1, 1, 1, 1, 0, S_1, S_2)$ will be equivalent to one of them. However these codes~\cite{GulHar2006} were found by an exhaustive computer search and do not give an algebraic structure.
\end{ex}


\medskip

\begin{ex}

As $4$-cyclotomic cosets produce good binary DDC codes, it is
natural to consider $9$-cyclotomic cosets to construct good
ternary DDC codes. In fact, we find a new ternary self-dual code
which has the best known parameters $[76, 38, 18]$. The code is
$\mathcal{B}_{37}(1,1,1,2,0,2, S_1, S_2)$ over $\F_3$ with the following
splitting.
\begin{equation*}
S_1=\{1, 9, 7, 26, 12, 34, 10, 16, 33, 2, 18, 14, 15, 24, 31, 20,
32, 29\}, \mu_{-1}.
\end{equation*}

We remark that with the splitting of quadratic residues and
quadratic nonresidues of $n=37$, there is no self-dual code. It is
known~\cite{GabOtm} that there is one $[76, 38, 18]$ ternary
self-dual code, denoted by $(f_1;1,35)$. We have checked that
$\mathcal{B}_{37}(1,1,1,2,0,2, S_1, S_2)$ is not equivalent to $(f_1;1,35)$.
More precisely, our code has the automorphism group of order $1332$
and the number of minimum codewords $A_{18} = 79032$ while
$(f_1;1,35)$ has the automorphism group of order $76$ and $A_{18} =
71136$.
\end{ex}

\begin{thm}
There are at least two ternary self-dual $[76,38,18]$ codes.
\end{thm}

\section{Conclusion}

We have introduced a subclass of double circulant codes, called
duadic double circulant codes. This is a natural generalization
of quadratic double circulant codes~\cite{Gab}. We have
constructed many interesting linear (self-dual) codes over
$\mathbb F_2$, $\mathbb F_3$, $\mathbb F_4$, $\mathbb F_5$ and
$\mathbb F_7$. As shown from our tables, many of our
codes have good minimum distances. In particular, we have found a
new self-dual ternary $[76,38,18]$ code, which is not equivalent to
the previously known code with the same parameters in~\cite{GabOtm}.

\section*{Acknowledgment}

 S. Han was supported by the Basic Science Research Program through the National Research Foundation of Korea (NRF), which is supported by the Ministry of Education, Science and Technology (2010-0007232).
 J.-L. Kim was partially supported by the Project Completion Grant (year 2011-2012) at the University of Louisville.




\begin{table}
\caption{Duadic double circulant codes over $\F_2$}
\label{tab:F2}
\centering
\begin{tabular}{|c|c|c|c|c|c|c|c|l|}                               \hline
 $n$ & cl    & SD(I) &SD(II) & NSD    & O.SD(I)  & O.SD(II) & O.Linear   & Comment           \\  \hline
 3   &  6    &  2    &       &   3    &   2      &          &      3     &   O(I), O(L)  \\  \hline
 3   &  8    &  2    &  4    &   3    &   2      &    4     &      4     &   O(I), O(II), O(L) \\  \hline
 5   &  10   &  2    &       &   4    &   2      &          &      4     &   O(I), O(L)     \\  \hline
 5   &  12   &  4    &       &   4    &   4      &          &      4     &   O(I), O(L)     \\  \hline
 7   &  14   &  2    &       &   4    &   4      &          &      4     &   O(L)     \\  \hline
 7   &  16   &  2    &  4    &   4    &   4      &    4     &      5     &   O(II)     \\  \hline
 9   &  18   &  4    &       &   4    &   4      &          &      6     &   O(I)     \\  \hline
 9   &  20   &  4    &       &   4    &   4      &          &      6     &   O(I)     \\  \hline
 11  & 22    &  6    &       &   7    &   6      &          &      7     &   O(I), O(L)     \\  \hline
 11  & 24    &  2    &  8    &   7    &   6      &    8     &      8     &   O(II), O(L)     \\  \hline
 13  & 26    &  2    &       &   7    &   6      &          &      7     &   O(L)      \\  \hline
 13  & 28    &  4    &       &   8    &   6      &          &      8     &   O(L)      \\  \hline
 15  & 30    &  2    &       &   8    &   6      &          &      8     &   O(L)      \\  \hline
 15  & 32    &  2    &  4    &   8    &   8      &    8     &      8     &   O(L)      \\  \hline
 17  & 34    &  2    &       &   8    &   6      &          &      8     &   O(L)     \\  \hline
 17  & 36    &  4    &       &   8    &   8      &          &      8     &   O(L)  \\  \hline
 19  & 38    &  8    &       &   7    &   8      &          &      8-9   &   O(I), B(L)     \\  \hline
 19  & 40    &  2    &  8    &   8    &   8      &    8     &      9-10  &   O(II)     \\  \hline
 29  & 58    &  2    &       &   12   &   10     &          &      12-14 &   O(L)   \\  \hline
 29  & 60    &  4    &       &   12   &   12     &          &      12-14 &   O(L)   \\  \hline
 33  & 66    &  12   &       &  12    &   12     &          &      12-16 &   O(I), B(L)    \\  \hline
 33  & 68    &  12   &       &  12    &   12     &          &      13-16 &   O(I)   \\  \hline
 41  & 82    &  2    &       &   14   &   14-16  &          &      14-20 &   B(L) \\  \hline
\end{tabular}
\end{table}

\begin{table}
\caption{Duadic double circulant codes over $\F_3$}
\label{tab:F3}
\centering
\begin{tabular}{|c|c|c|c|c|c|c|}                               \hline
 $n$ & cl    & SD   & NSD   & O.SD    & O.Linear     & Comment       \\  \hline
 3   &  6    &      &  3    &         &    3         &   O(L)        \\  \hline
 3   &  8    &     &  4    &   3     &    4         &   O(L)        \\  \hline
 5   &  10   &      &  5    &         &    5         &   O(L)        \\  \hline
 5   &  12   &  6   &  5    &   6     &    6         &   O(S), O(L)  \\  \hline

 7   &  16   &  6   &  6    &   6     &    6         &   O(S), O(L)  \\  \hline

 11  & 22    &      &  8    &         &    8         &   O(L)        \\  \hline
 11  & 24    &  9   &  8    &   9     &    9         &   O(S), O(L)  \\  \hline
 13  & 26    &      &  8    &         &    8-9       &   B(L)        \\  \hline

 17  & 34    &      &  11   &         &    11-12     &   B(L)        \\  \hline
 17  & 36    &  12  &  11   &   12    &    12        &   O(S), O(L)  \\  \hline
 19  & 38    &      &  11   &         &    11-13     &   O(L)        \\  \hline
 19  & 40    &  12  &  11   &   12    &    12-14     &   O(S), B(L)  \\  \hline

 23  & 46    &      &  14   &         &    14-15     &   B(L)        \\  \hline
 23  & 48    &  15  &  14   &   15    &    15-16     &   O(S), B(L)  \\  \hline

 29  & 58    &      &  17   &         &    17-19     &   B(L)        \\  \hline
 29  & 60    &  18  &  17   &   18    &    18-20     &   O(S), B(L)  \\  \hline
 31  & 62    &      &  17   &         &    17-20     &   B(L)        \\  \hline
 31  & 64    &  18  &  17   &   18    &    18-21     &   O(S), B(L)  \\  \hline

 41  & 82    &      &  20   &         &    20-26     &   B(L)        \\  \hline
\end{tabular}
\end{table}

\begin{table}
\caption{Duadic double circulant codes over $\F_4$}
\label{tab:F4}
\centering
\begin{tabular}{|c|c|c|c|c|c|c|c|c|c|}                               \hline
 $n$ & cl    & SD(E)  & SD(H) & NSD    & O.SD(E)  & O.SD(H)   & O.Linear & Comment            \\  \hline
 3   &  6    &  2    &   4    &  3     &    3     &    4      &    4     &   O(H), O(L)       \\  \hline
 3   &  8    &  4    &   4    &  4     &    4     &    4      &    4     &   O(E), O(H), O(L) \\  \hline
 5   &  10   &  2    &   4    &  4     &    4     &    4      &    5     &   O(H)             \\  \hline
 5   &  12   &  4    &   4    &  4     &    6     &    4      &    6     &   O(H)             \\  \hline
 7   &  14   &  6    &   4    &  6     &    6     &    6      &    6     &   O(E), O(L)       \\  \hline
 7   &  16   &  6    &   4    &  6     &    6     &    6      &    7     &   O(E)             \\  \hline
 11  & 22    &  6    &   8    &  7     &    8     &    8      &    8-9   &   O(H), B(L)       \\  \hline
 11  & 24    &  8    &   8    &  8     &    8-10  &    8      &    9  &   B(E), O(H)       \\  \hline

 17  & 34    &  2    &   10   &  11    &    10-12 &    10-12  &    11-13 &   B(H), B(L)       \\  \hline
 17  & 36    &  4    &   12   &  11    &    11-14 &    12-14  &    12-14 &   B(H), B(L)       \\  \hline
 23  & 46    &  14   &   4    &  14    &    14-16 &    14-16  &    14-17 &   B(E), B(L)       \\  \hline
 23  & 48    &  14   &   4    &  14    &    14-18 &    14-18  &    14-18 &   B(E), B(L)       \\  \hline
 31  & 62    &  16   &   4    &  16    &    16-23 &    18-22  &    18-23 &   B(E)             \\  \hline

\end{tabular}
\end{table}
\begin{table}
\caption{Duadic double circulant codes over $\F_5$}
\label{tab:F5}
\centering
\begin{tabular}{|c|c|c|c|c|c|c|}                               \hline
 $n$ & cl    & SD    & NSD    &   O.SD   & O.Linear & Comment\\  \hline
 3   &  6    &  4    &  4     &    4     &     4        &   O(S), O(L)     \\  \hline
 3   &  8    &  4    &  4     &    4     &     4        &   O(S), O(L)     \\  \hline
 5   &  10   &  2    &  5     &    4     &     5        &   O(L)     \\  \hline
 5   &  12   &  4    &  6     &    6     &     6        &   O(L)     \\  \hline
 7   &  14   &  6    &  6     &    6     &     6        &   O(S), O(L)     \\  \hline
 7   &  16   &  7    &  7     &    7     &     7        &   O(S), O(L)      \\  \hline

 11  & 22    &  7    &  8     &    8     &     8-10     &   B(L)     \\  \hline
 11  & 24    &  9    &  9     &    9     &     9-10     &   O(S), B(L)     \\  \hline

 17  & 34    &  4    &  11    &    11-12 &     11-14    &   B(L)     \\  \hline
 17  & 36    &  4    &  12    &    12-13 &     12-15    &   B(L)     \\  \hline
 19  & 38    &  12   &  12    &    12-14 &     12-16    &   B(S), B(L)     \\  \hline
 19  & 40    &  13   &  13    &    13-15 &     13-17    &   B(S), B(L)     \\  \hline

 23  & 46    &  14   &  14    &    14-20 &     14-20    &   B(S), B(L)      \\  \hline
 23  & 48    &  14   &  15    &    14-20 &     15-20    &   B(S), B(L)      \\  \hline

 29  & 58    &  16   &  17    &    16-24 &     18-24    &   B(S) \\  \hline
\end{tabular}
\end{table}

\begin{table}
\caption{Duadic double circulant codes over $\F_7$}
\label{tab:F7}
\centering
\begin{tabular}{|c|c|c|c|c|c|c|}                                       \hline
 $n$ & cl    & SD   & NSD   & O.SD & O.Linear  & Comment     \\  \hline
 3   &  6    &      &  4    &         &     4        &   O(L)      \\  \hline
 3   &  8    &  5   &  5    &    5    &     5        &   O(S), O(L)\\  \hline
 5   &  10   &      &  5    &         &     5        &   O(L)      \\  \hline
 5   &  12   &  6   &  6    &    6    &     6        &   O(S), O(L)\\  \hline

 7   &  16   &     &  7    &    7    &     7        &   O(L)      \\  \hline

 11  & 22    &      &  9    &         &     9-10     &   B(L)      \\  \hline
 11  & 24    &  9   &  9    &    9-11 &     10-11    &   B(S)      \\  \hline
 13  & 26    &      &  10   &         &     10-12    &   B(L)      \\  \hline
 13  & 28    &  10  &  10   &    11-13&     11-13    &   B(S)      \\  \hline

 17  & 34    &      &  12   &         &     12-15    &   B(L)      \\  \hline

 19  & 38    &      &  13   &         &     13-17    &   B(L)      \\  \hline

\end{tabular}
\end{table}

\end{document}